\documentclass{article}
\usepackage[utf8]{inputenc}

\usepackage{graphicx}
\usepackage{time}

\usepackage{amsthm} 

\newtheorem{theorem}{Theorem}

\newtheorem{lemma}[theorem]{Lemma}

\title{A Gray Code of Ordered Trees}

\author{
 Shin-ichi Nakano\\  Gunma University
}

\begin{document}

\maketitle


\vskip 10mm

\noindent
{\bf Abstract}
A combinatorial Gray code for a set of combinatorial objects
is a sequence of
all combinatorial objects in the set 
so that 
each object is derived from the preceding object
by changing a small part.

In this paper 
we design a Gray code for ordered trees with $n$ vertices
such that
each ordered tree is derived from the preceding ordered tree
by removing a leaf then appending a leaf elsewhere.
Thus 
the change is just remove-and-append a leaf,
which is the minimum.

\section{Introduction}
\label{sec:intro}

A classical Gray code for $n$-bit binary numbers
is a sequence of all $n$-bit binary numbers 
so that 
each number is derived from the preceding number
by changing exactly one bit.
A combinatorial Gray code for a set of combinatorial objects
is a sequence of
all combinatorial objects in the set
so that 
each object is derived from the preceding object
by changing a small (constant) part.

When we generate all combinatorial objects
and the number of such objects is huge
if we can compute them as a combinatorial Gray code
then 
we can output (or store)
each object as a small size of the difference from the preceding object
and we may compute each object in a constant time.
Also,
when we repeatedly solve some problem for a class of objects,
a solution for an object
may help to compute a solution 
for a similar successive object.
See surveys for combinatorial Gray codes  \cite{S97, M22}.

For binary trees
with $n$ vertices
one can generate all binary trees 
so that 
each binary tree is derived from the preceding binary tree
by a rotation operation at a vertex \cite{L87, L93}.
The number of change of edges 
in a rotation operation
is three \cite[p9]{K06}.
Also one can generate all binary trees 
with $n$ vertices
so that 
each tree is derived from the preceding tree
by removing a subtree and place it elsewhere \cite[Exercise 25]{K06}.
However the levels of many vertices may be changed,
where 
the level of a vertex is the number of vertices
on the path from the vertex to the root.


In this paper 
we design a Gray code for ordered trees with $n$ vertices
such that
each ordered tree is derived from the preceding ordered tree
by removing a leaf then appending a leaf elsewhere.
Thus 
the change is just remove-and-append a leaf,
which is the minimum,
and  other vertices remain as they were
including their levels.
Our Gray code is based on a tree structure 
among the ordered trees.


The remainder of this paper is organized as follows.
Section 2 gives
some definitions and basic lemmas.
In Section 3 we  design our algorithm
to construct a Gray code for the ordered trees with $n$ vertices.
Finally Section 4 is a conclusion.

\section{Preliminaries}

{\it A tree} is a connected graph with no cycle.
{\it A rooted tree} is a tree 
with a designated vertex as {\it the root}.
{\it The level of a vertex} $v$ in a rooted tree is the number of vertices
on the path from $v$ to the root.
The level of the root
is $1$.
For each vertex $v$ except the root
if the neighbor vertex of $v$ on the path from $v$ to the root
is $p$ then 
$p$ is {\it the parent} of $v$ and $v$ is {\it a child} of $p$.
The root has no parent.
In this paper 
we always draw each child vertex 
below its parent.
A vertex with no child is called {\it a leaf}.
{\it An ordered tree} is a rooted tree 
in which the left-to-right order of child vertices 
of each vertex is defined.
The number of ordered trees with exactly $n+1$ vertices
is known as the $n$-th Catalan number $_{2n}C_n / (n+1)$ \cite[p12]{K06}.

Given an ordered tree $T$,
let $P_r(T)=(v_0,v_1,\cdots, v_k)$
be the path from the root $v_0$ to a leaf $v_k$
such that, 
for each $i=1,2,\cdots,k$,
$v_i$ is the rightmost child of $v_{i-1}$.
$P_r(T)$ is called 
{\it the rightmost path} of $T$ and
$v_k$ is called
{\it the rightmost leaf} of $T$.
The number of edges in $P_r(T)$
is denoted by $rpl(T)$.

For an ordered tree $T$
if the rightmost child of the root has exactly one child
as a leaf
then we say
$T$ has {\it the pony-tail}.

For two distinct ordered trees $T$ and $T'$,
if $T'$ is derived from $T$
by appending a new leaf as the rightmost leaf
then removing other leaf,
then we say 
$T$ {\it is copying} $T'$ (at level $rpl(T')$).
%
When $T$ is copying $T'$
if the parent of the rightmost leaf of $T'$
has two or more child vertices
then $rpl(T)\ge rpl(T')$ holds,
otherwise, 
the parent of the rightmost leaf of $T'$
has exactly one child vertex, which is the rightmost leaf,
and
$rpl(T)= rpl(T')-1$ holds.
So
if $T$ is copying $T'$,
$rpl(T)=1$ and $rpl(T')>1$ then $T'$ has the pony-tail.


Let $S_k$ be the set of the ordered trees 
with exactly $k$ vertices.
In this paper we design,
for each $k=1,2,\cdots, n$, a combinatorial Gray code
for $S_k$,
that is a sequence of all ordered trees in $S_k$
such that
each ordered tree is derived from
the preceding ordered tree
by removing a leaf
then appending a leaf elsewhere.
We call the change
{\it delete-and-append a leaf}.

For an ordered tree $T$ with $n\ge 2$ vertices
let $p(T)$ be the ordered tree 
derived from $T$ by removing the rightmost leaf.
We say $p(T)$ is {\it the parent} of $T$,
and $T$ is a child of $p(T)$.
For any ordered tree $T$ in $S_n$
if we repeatedly compute the parent of the derived ordered tree
we obtain the sequence $T,p(T), p(p(T)),\cdots $ of ordered trees,
which ends with the trivial ordered tree 
consisting of exactly one vertex.
We call the sequence {\it the removing sequence of $T$} \cite{N02}.

\begin{figure}[htb]
  \begin{center}
     \includegraphics[width=13cm,pagebox=cropbox,clip]{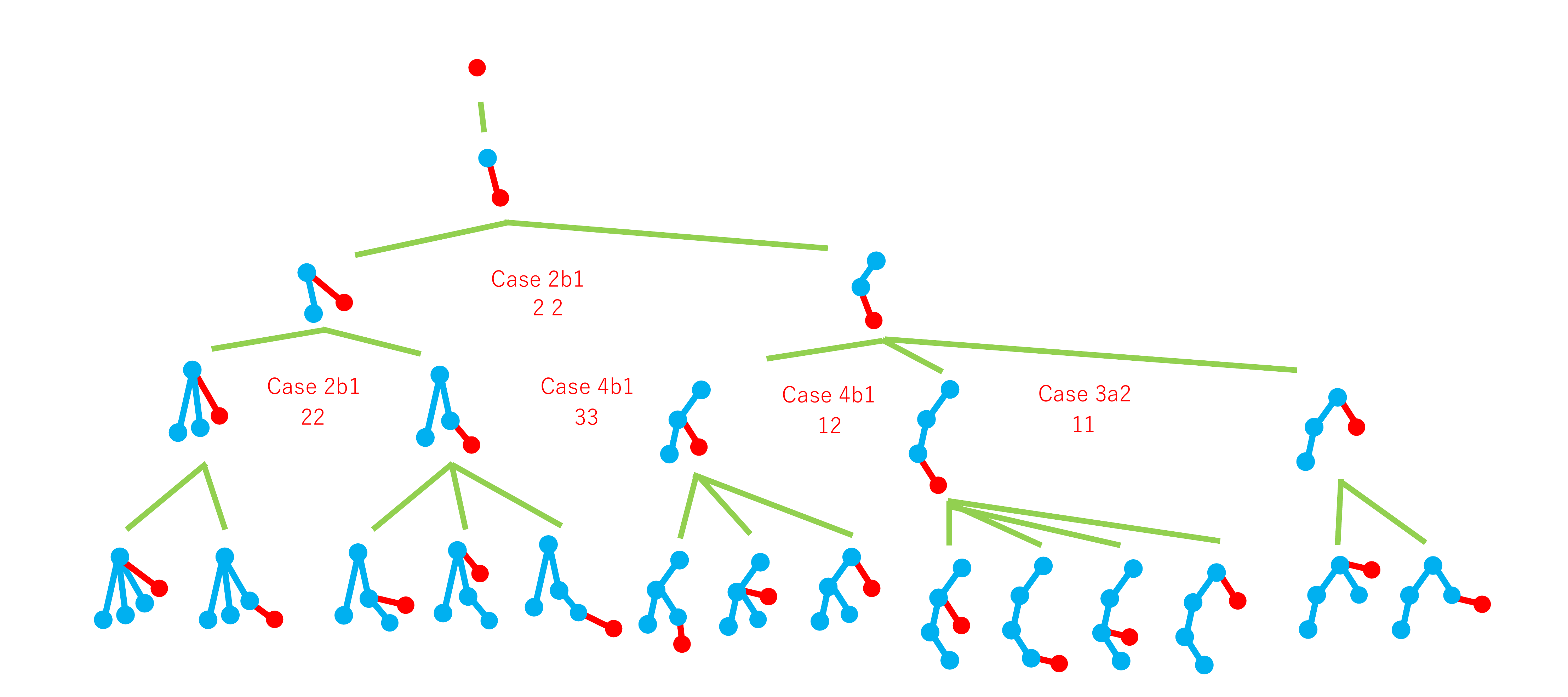}
  \end{center}
  \caption{
   The family tree $F_n$ of $S_n$.
  }
  \label{fig:ftree}
\end{figure}

By merging the removing sequences of the ordered trees in $S_n$
one can obtain an (unordered) tree $F_n$
of ordered trees \cite{N02} (See an example for $n=5$ in Fig. \ref{fig:ftree})
in which 
the root corresponds to the trivial ordered tree with exactly one vertex,
each vertex at level $k$ corresponds 
to some ordered tree in $S_k$,
and
each edge corresponds to some ordered tree and its parent.
We call the tree  {\it the family tree}.
Note that we have not decide yet 
the left-to-right order of the child ordered trees
of each order tree in $F_n$.
We have the following three lemmas.

\begin{lemma}\label{le:ftree1}
There is a bijection
between 
the ordered trees in $S_k$ and 
the vertices at level $k$
in $F_n$.
\end{lemma}
\begin{proof} 
Given an ordered tree $T$ with exactly $k$ vertices,
by repeatedly appending a new leaf 
as the rightmost child of the root,
one can obtain a descendant tree
$T'\in S_n$ in $F_n$.
Thus every order tree in $S_k$ 
appears in the removing sequence of some tree in $S_n$
and so
corresponds to a vertex 
at level $k$ in $F_n$.

Clearly every vertex at level $k$
in $F_n$
corresponds to an ordered tree with exactly $k$ vertices.
\end{proof} 

\begin{figure}[htb]
  \begin{center}
     \includegraphics[width=9cm,pagebox=cropbox,clip]{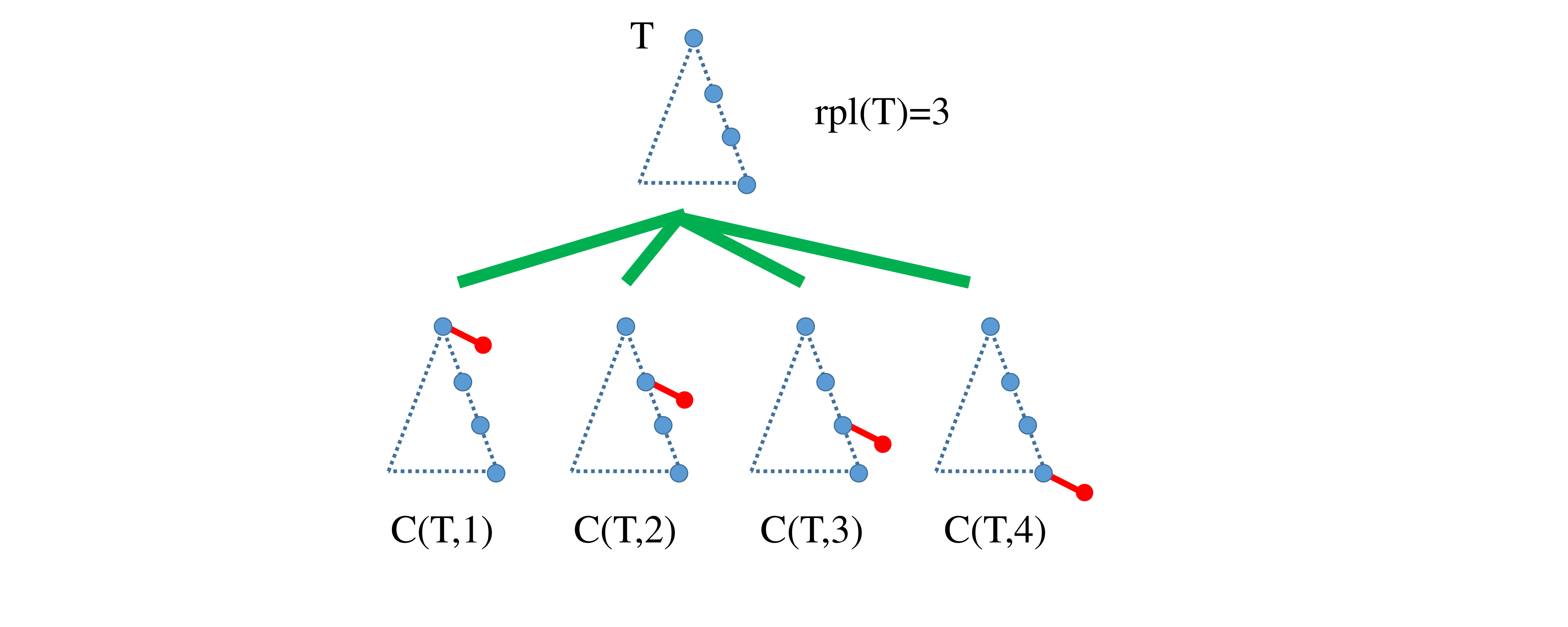}
  \end{center}
  \caption{
   An illustration for Lemma 
   \ref{le:ftree2}.
  }
  \label{fig:child}
\end{figure}

\begin{lemma}\label{le:ftree2}
Let $T$ be an ordered tree in $S_k$ 
with $k<n$.
$T$ has $rpl(T)+1$ child ordered trees in $F_n$.
\end{lemma}
\begin{proof} 
For each $i=1,2,\cdots,rpl(T)+1$,
by appending a new leaf as the rightmost child leaf
of the vertex on $P_r(T)$ at level $i$,
one can obtain a distinct child ordered tree.
See Fig.\ref{fig:child}.
\end{proof} 

We denote by $C(T,i)$
the child ordered tree of $T$
derived from $T$
by appending a new leaf as the rightmost child leaf
of the vertex on $P_r(T)$ at level $i$.
Thus $rpl(C(T,i))=i$.

Thus, by Lemma \ref{le:ftree2},
every ordered tree $T$ in $S_k$ with $k<n$
except the ordered tree with exactly one vertex
has 
two or more child ordered trees in $F_n$
since $rpl(T)\ge 1$.
Clearly the ordered tree with exactly one vertex
has 
exactly one child ordered tree in $F_n$.

\begin{lemma}\label{le:ftree3}
Any ordered tree is derived from its sibling ordered tree
by delete-and-append a leaf.
\end{lemma}
\begin{proof} 
Any ordered tree is derived from its sibling ordered tree
by deleting the rightmost leaf 
then appending a leaf as the rightmost leaf
at the suitable level.
\end{proof} 

In this paper we show that
by suitably defining the left-to-right order of child ordered trees
of each ordered tree in $F_n$,
we can define an ordered tree $F^O_n$
such that, for each $k$,
a Gray code for $S_k$ is appeared
as 
the left-to-right sequence of 
the ordered trees corresponding to 
the vertices at level $k$ of $F^O_n$.
Thus a Gray code for $S_n$ is appeared
as 
the left-to-right sequence of 
the ordered trees corresponding to 
the leaves of $F^O_n$.
See an example for $n=5$ in Fig. \ref{fig:ftree}.

\section{Algorithm}

In this section we design a Gray code for $S_k$
for each $k=1,2,\cdots,n$,
where $S_k$ is the set of the ordered trees
with exactly $k$ vertices.

\noindent
{\bf Induction on levels}
We proceed by induction on levels.
Let $F_k$ be the subtree of $F_n$
induced by 
$S_1\cup S_2\cup\cdots\cup S_k$.
The Gray code for $S_1$ is trivial and unique
since $|S_1|=1$.
Simillar for $S_2$ since $|S_2|=1$.
Assume that,
for an integer $k<n$,
we have
defined a left-to-right order of child ordered trees
of each ordered tree 
in $S_1\cup S_2\cup\cdots\cup S_{k-1}$,
we have obtained
an ordered tree $F^O_k$ corresponding to $F_k$,
and 
we have constructed a Gray code for $S_k$
as 
the left-to-right sequence of 
the ordered trees corresponding to the leaves of $F^O_k$.
Then
we are going to define
a left-to-right order of child ordered trees
of each ordered tree in $S_k$
so that
it extends $F^O_k$ to
an ordered tree $F^O_{k+1}$ and
a Gray code for $S_{k+1}$ is appeared
as 
the left-to-right sequence of 
the ordered trees at the leaves of $F^O_{k+1}$.

\noindent
{\bf Basic strategy of algorithm}
%
Let ($T_1,T_2,\cdots)$ be our Gray code for $S_k$.
We are going to define 
a left-to-right order of child ordered trees 
of each $T_i$ in $S_k$,
then we obtain a sequence of ordered trees,
 which is a Gray code for $S_{k+1}$, 
say $(T'_1, T'_2,\cdots )$.

%

If two consecutive ordered trees $T'_j$ and $T'_{j+1}$ in the sequence
are siblings in $F^O_{k+1}$,
then one can be derived from the other
by delete-and-append a leaf by Lemma \ref{le:ftree3}.
However if
two consecutive ordered trees $T'_j$ and $T'_{j+1}$
are not siblings in $F^O_{k+1}$,
that is,
$T'_j$ is the rightmost child ordered tree of $T_i$ and
$T'_{j+1}$ is the leftmost child ordered tree of $T_{i+1}$ 
for some $i$,
then we have several cases to consider. 
We have the following lemma 
for $T_i$ and $T_{i+1}$.

\begin{figure}[htb]
  \begin{center}
     \includegraphics[width=10cm,pagebox=cropbox,clip]{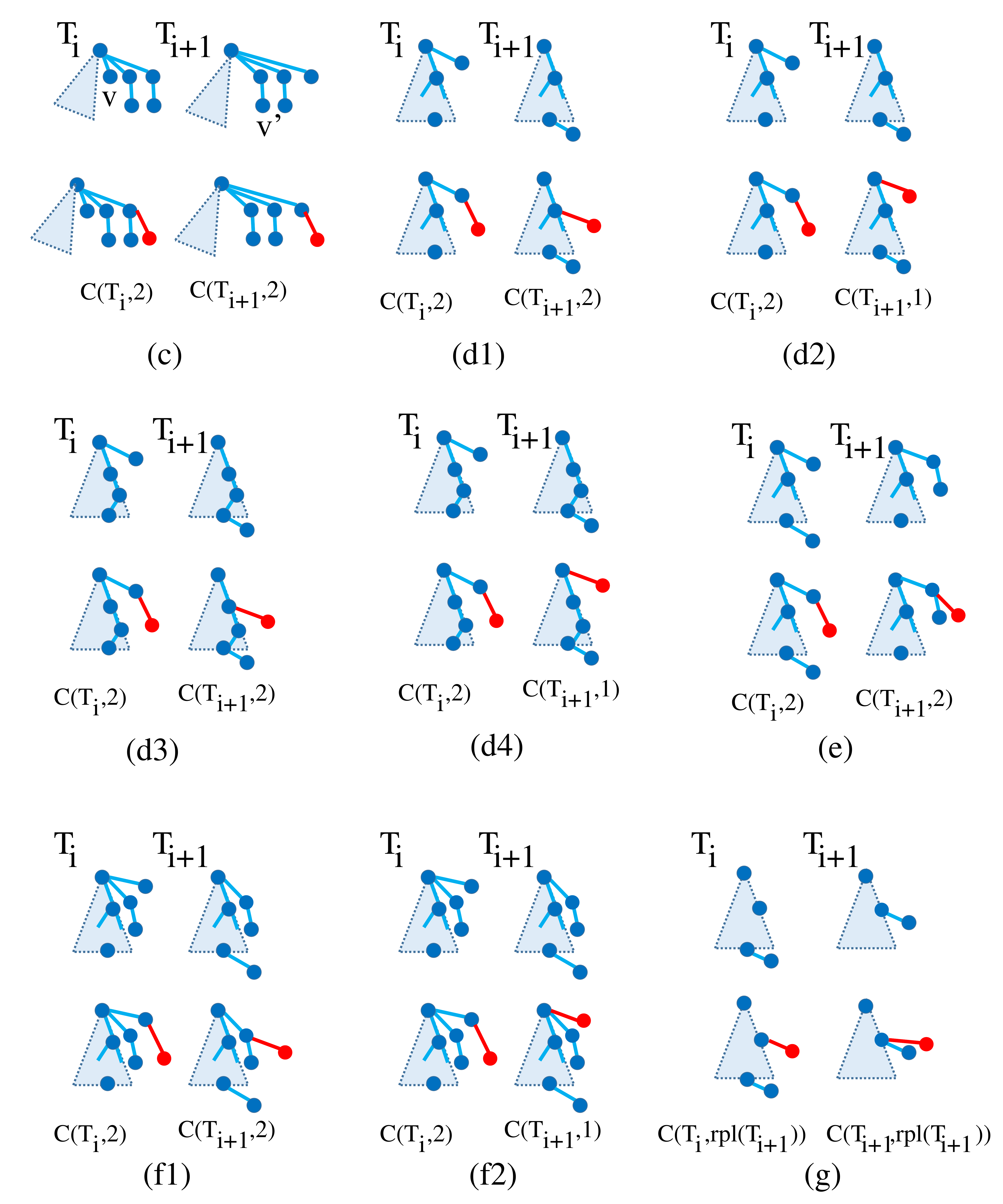}
  \end{center}
  \caption{
   Illustration for Lemma \ref{le:cases}.
  }
  \label{fig:cases}
\end{figure}

\begin{lemma}\label{le:cases}
Assume that $T_i$ can be derived from $T_{i+1}$ 
by delete-and-append a leaf.
Then the followings are hold.

\begin{enumerate}
\item[{\rm (a)}]
$C(T_i,1)$ can be derived from $C(T_{i+1},1)$ 
by delete-and-append a leaf.

\item[{\rm (b)}]
If $rpl(T_i)=rpl(T_{i+1})=1$, then
$C(T_i,2)$ can be derived from $C(T_{i+1},2)$ 
by delete-and-append a leaf.

\item[{\rm (c)}]
If
$rpl(T_i)$ has the pony-tail,
$rpl(T_{i+1})=1$, 
$T_{i}$ is copying $T_{i+1}$ at level $1$ and
$T_{i+1}$ is copying $T_i$ at level $2$, 
then
$C(T_i,2)$ can be derived from $C(T_{i+1},2)$ 
by delete-and-append a leaf.

\item[{\rm (d)}]
If
$rpl(T_i)=1$, $rpl(T_{i+1})>1$, and
$T_{i+1}$ has no pony-tail
(so $T_{i+1}$ is copying $T_i$ at level $1$), then
$C(T_i,2)$ can not be derived from $C(T_{i+1},2)$ 
by delete-and-append a leaf
(See Fig.\ref{fig:cases} (d1) and (d3)),
however
$C(T_i,2)$ can be derived from $C(T_{i+1},1)$ 
by delete-and-append a leaf.
(See Fig.\ref{fig:cases} (d2) and (d4).)


\item[{\rm (e)}]
If
$rpl(T_i)=1$, $rpl(T_{i+1})>1$, 
$T_{i+1}$ has the pony-tail, and
$T_{i}$ is copying $T_{i+1}$ at level $2$,
then
$C(T_i,2)$ can be derived from $C(T_{i+1},2)$ 
by delete-and-append a leaf.
(See Fig. \ref{fig:cases} (e).)

\item[{\rm (e')}]
If
$rpl(T_i)>1$, $rpl(T_{i+1})=1$, 
$T_{i}$ has the pony-tail, and
$T_{i+1}$ is copying $T_{i}$ at level $2$,
then
$C(T_{i+1},2)$ can be derived from $C(T_{i},2)$ 
by delete-and-append a leaf.

\item[{\rm(f)}]
If
$rpl(T_i)=1$, $rpl(T_{i+1})>1$, 
$T_{i+1}$ has the pony-tail,
$T_{i+1}$ is copying $T_i$ at level $1$,
then
$C(T_i,2)$ can not be derived from $C(T_{i+1},2)$ 
by delete-and-append a leaf
(See Fig.\ref{fig:cases} (f1)),  however
$C(T_i,2)$ can be derived from $C(T_{i+1},1)$ 
by delete-and-append a leaf.
(See Fig.\ref{fig:cases} (f2).)


\item[{\rm (g)}]
If
$rpl(T_i)\ge rpl(T_{i+1})\ge 2$, 
then
$C(T_i,rpl(T_{i+1}))$ can be derived from $C(T_{i+1},rpl(T_{i+1}))$ 
by delete-and-append a leaf.
(See Fig.\ref{fig:cases} (g).)

If 
$rpl(T_i)=rpl(T_{i+1})\ge 2$, 
then
$C(T_i,rpl(T_i))$ can be derived from $C(T_{i+1},rpl(T_i))$ 
by delete-and-append a leaf, and
$C(T_i,rpl(T_i)+1)$ can be derived from $C(T_{i+1},rpl(T_i)+1)$ 
by delete-and-append a leaf.

\item[{\rm (g')}]
If
$rpl(T_{i+1})\ge rpl(T_{i})\ge 2$, 
then
$C(T_{i+1},rpl(T_{i}))$ can be derived from $C(T_{i},rpl(T_{i}))$.
Also
if $rpl(T_{i+1}) > rpl(T_{i})\ge 2$,
then
$C(T_{i},1)$ can be derived from $C(T_{i+1},rpl(T_{i}))$ 
by delete-and-append a leaf.
\end{enumerate}

\end{lemma} 
\begin{proof}
{(a)} 
{(b)} 
We have the following two cases. 
Case 1: $T_i$ is derived from $T_{i+1}$ by removing the rightmost leaf 
then appending a new leaf elsewhere.
Case 2: $T_i$ is derived from $T_{i+1}$ by removing a leaf 
which is not the rightmost leaf
then appending a new leaf elsewhere.
For both cases the claim holds.

\noindent
{(c)}
Assume that 
$T_{i+1}$ is derived from $T_{i}$ by
appending the rightmost leaf at level $1$
then 
deleting a leaf $v$
(since $T_i$ is copying $T_{i+1}$),
and
$T_{i}$ is derived from $T_{i+1}$ by
appending the rightmost leaf at level $2$
then
deleting a leaf $v'$
(since $T_{i+1}$ is copying $T_{i}$).

We can show that exactly one of $v$ or $v'$ is a child of the root, as follows.
If $v$ is a child of the root of $T_{i}$ and
$v'$ is a child of the root of $T_{i+1}$
then, since $T_{i}$ is copying $T_{i+1}$,
the degree of the root of $T_{i}$ is equal to the degree of the root of $T_{i+1}$,
and,
since $T_{i+1}$ is copying $T_{i}$,
the degree of the root of $T_{i+1}$ minus $1$ 
is equal to  the degree of the root of $T_{i}$,
a contradiction.
Also
if $v$ is not a child of the root of $T_{i}$ and
$v'$ is not a child of the root of $T_{i+1}$
then, since $T_{i}$ is copying $T_{i+1}$,
the degree of the root of $T_{i}$ plus $1$ is equal to 
the degree of the root of $T_{i+1}$,
and,
since $T_{i+1}$ is copying $T_{i}$,
the degree of the root of $T_{i+1}$ is the degree of the root of $T_{i}$,
a contradiction.
Thus exactly one of $v$ or $v'$ is a child of the root.

Assume first that $v$ is a child of the root of $T_i$.
Let $x_1,x_2,\cdots,x_d$ be the child vertices of the root in $T_{i}$ except $v$ in right-to-left order,
and $y_1,y_2,\cdots,y_{d+1}$ the child vertices of the root in $T_{i+1}$ 
in right-to-left order.
Since $T_{i}$ is copying $T_{i+1}$,
after removing $v$ from $T_{i}$,
the subtrees rooted at $x_1,x_2,\cdots, x_d$ are identical to 
the subtrees rooted at $y_2,y_3,\cdots, y_{d+1}$, respectively.
Also 
since $T_{i+1}$ is copying $T_{i}$,
after removing $v'$ from $T_{i+1}$,
the subtrees rooted at $y_2,y_3,\cdots, y_{d+1}$ except one 
(corresponding to the trivial subtree rooted at $v$) are identical to 
the subtrees rooted at $x_2,x_3,\cdots, x_d$, respectively.
If $v'$ belong to a subtree rooted at, say $y_j$, 
then, since $T_i$ is copying $T_{i+1}$,
the subtree rooted at $x_{j-1}$ is identical to the subtree rooted at $y_j$
and also, since $T_{i+1}$ is copying $T_i$,
after removing $v'$ from the subtree rooted at $y_j$,
if it is identical to the subtree rooted at $x_{j-1}$, then,
a contradiction.
Thus $v'$ belong to the subtree corresponding to the subtree rooted at $v$,
that is $v'$ is the only child 
of a child (corresponding to $v$) 
of the root.
See Fig.\ref{fig:cases} (c).
Now $C(T_i,2)$ is derived from $C(T_{i+1},2)$ by delete-and-append a leaf.

Simillar for the case 
where $v'$ is a child of the root of $T_{i+1}$.

\noindent
{(d)} 
Since $T_{i+1}$ has no pony-tail,
either (Case 1)
the rightmost child vertex of the root of $T_{i+1}$
has two or more child vertices
(See Fig.\ref{fig:cases} (d1)), or
(Case 2) 
the rightmost child vertex 
of the rightmost child vertex 
of the root of $T_{i+1}$
has one or more child vertices
(See Fig.\ref{fig:cases} (d3)).
Since $rpl(T_i)=1$
the rightmost child vertex of the root of $T_i$
has no child vertex.
For Case 1,
the rightmost child vertex of the root of $C(T_{i+1},2)$
has three or more child vertices,
while
the rightmost child vertex of the root of $C(T_i,2)$
has exactly one child vertex.
Thus $C(T_i,2)$ can not be derived from $C(T_{i+1},2)$
by delete-and-append a leaf.
See Fig.\ref{fig:cases} (d1).
For Case 2
we need to remove at least two vertices and append at least two vertices
to obtain $C(T_i,2)$ from $C(T_{i+1},2)$.
Thus $C(T_i,2)$ can not be derived from $C(T_{i+1},2)$
by delete-and-append a leaf.
See Fig.\ref{fig:cases} (d3).
However
$C(T_{i},2)$ can be derived from $C(T_{i+1},1)$
by delete-and-append a leaf.
See Fig.\ref{fig:cases} (d2) and (d4).

\noindent
{(e)} 
See Fig.\ref{le:cases} (e).

\noindent
{(e')} 
Similar to (e).

\noindent
{(f)}
See Fig.\ref{fig:cases} (f1) and (f2).

\noindent
{(g)}
See Fig.\ref{fig:cases} (g).

\noindent
{(g')} Similar to (g).
\end{proof}

\noindent
{\bf Step of algorithm}
Let $(T_1,T_2,\cdots )$ be a Gray code 
for $S_k$ corresponding to the leaves of $F^O_k$
and we are going to 
define a left-to-right order of child ordered trees 
of each ordered tree in $S_k$
and
construct a Gray code $(T'_1,T'_2,\cdots )$
for $S_{k+1}$ corresponding to the leaves of $F^O_{k+1}$.
When we start step $i$
assume that we have already defined 
the left-to-right order of the child ordered trees 
of $T_1,T_2,\cdots,T_{i-1}$ and
the leftmost child ordered tree of $T_i$,
and
in step $i$
we are going to define the left-to-right order of 
the child ordered trees of $T_i$ except the leftmost one,
and 
the leftmost child ordered trees of $T_{i+1}$.
See Fig.\ref{fig:step}.
The part we are going to define 
in the current step $i$ is depicted as a grey rectangle.
We proceed with several cases
based on $rpl(T_i),rpl(T_{i+1})$ and the leftmost child of $T_i$,
as explained later.
\begin{figure}[htb]
  \begin{center}
     \includegraphics[width=9cm,pagebox=cropbox,clip]{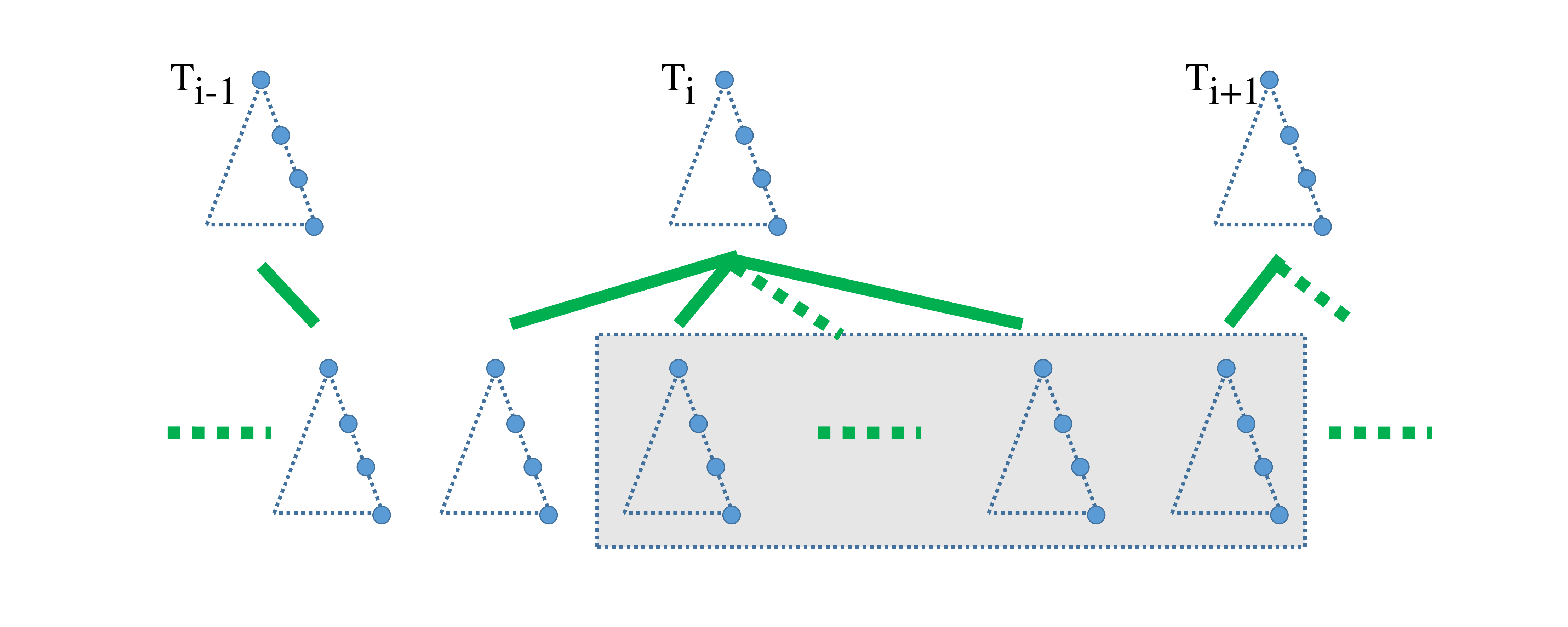}
  \end{center}
  \caption{
  An illustration for step $i$ of the algorithm.
  }
  \label{fig:step}
\end{figure}
%

\noindent
{\bf Loop invariants}
%

Our algorithm satisfies the following two conditions 
at each step $i$. (Note that (co1) is independent of $i$.)

\begin{enumerate}
\item[(co1)] 
For consecutive three ordered trees $T_{u-1},T_u,T_{u+1}$
at level $k$,
if 
$rpl(T_{u-1}) = rpl(T_{u+1}) = 1$ and
$rpl(T_{u}) >1 $
then 
$T_u$ has the pony-tail
and 
$T_{u+1}$ is copying $T_u$ at level $2$.
Also
if 
$rpl(T_{u-1}) = rpl(T_{u+1})\ge 2$ then
$rpl(T_{u-1}) > rpl(T_u) $.

\item[(co2)] 
For consecutive three ordered trees  $T'_{u'-1},T'_{u'},T'_{u'+1}$ at level $k+1$
with $u'+1\le i'$, 
where $T'_{i'}$ is the leftmost child ordered tree of $T_i$,
if
$rpl(T'_{u'-1}) = rpl(T'_{u'+1}) = 1$ and
$rpl(T'_{u'}) >1 $
then
$T'_{u'}$ has the pony-tail
and
$T'_{u+1}$ is copying $T'_u$ at level $2$.
Also
if 
$rpl(T'_{u'-1}) = rpl(T'_{u'+1})\ge 2$ then
$rpl(T'_{u'-1}) > rpl(T_u') $.

\end{enumerate}

\begin{figure}[htb]
  \begin{center}
     \includegraphics[width=14cm,pagebox=cropbox,clip]{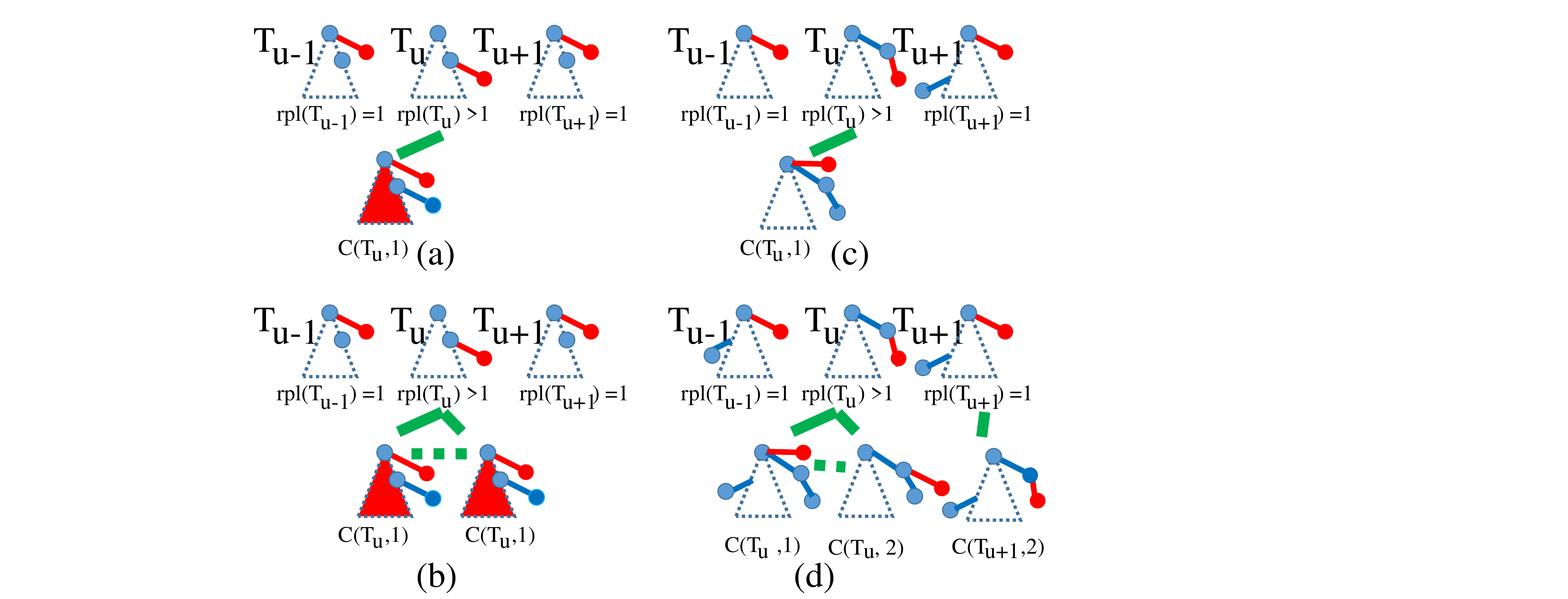}
  \end{center}
  \caption{
     Illustrations for the loop invariants.
  }
  \label{fig:invariant}
\end{figure}

The intuitive reason why we need those condition is as follows.

\noindent
Assume that  there are $T_{u-1},T_u,T_{u+1}$
with
$rpl(T_{u-1}) = rpl(T_{u+1}) = 1$,
$rpl(T_{u}) >1 $,
$T_u$ has no pony-tail, and
$C(T_u,1)$ is the leftmost child of $T_u$
(see Fig.\ref{fig:invariant}(a)),
and 
if we try to set $C(T_u,1)$
at the rightmost child of $T_u$,
then we  fail to construct a Gray code for $S_{k+1}$
since the same tree appear twice.
(See Fig.\ref{fig:invariant}(b).)
So our algorithm try to exclude any occurrence of
such consecutive three ordered trees.
Note that
even when 
$rpl(T_{u-1}) = rll(T_{u+1}) = 1$,
$rpl(T_{u}) >1 $ and
$C(T_u,1)$ is the leftmost child of $T_u$,
if $T_u$ has the pony-tail and 
$T_{u+1}$ is copying $T_u$
(see Fig.\ref{fig:invariant}(c)),
then
we can set $C(T_u,2)$
at the rightmost child of $T_i$ and
$C(T_{u+1},2)$ at the leftmost child of $T_{i+1}$
(by Lemma \ref{le:cases}(e'))
and we can proceed successfully.
(See an example in Fig.\ref{fig:invariant}(d).)

\vskip 3mm

\noindent
{\bf Algorithm}
First we 
set $C(T_1,1)$ as the leftmost child of $T_1$.

Assume that we have done each step $1,2,\cdots, i-1$.
Now we execute the next step $i$ of our algorithm
if $T_{i+1}$ exists.
(If $T_{i}$ is the last ordered tree
in the Gray code of $S_k$
then we order the remaining child of $T_{i}$ 
with decreasing order of $rpl$ from left to right.
See Fig. \ref{fig:ftree}.
Note that if $rpl(T_i)\ge 3$ 
then $C(T_i,1)$ never appear at the second leftmost child of $T_i$.)

We have the following four cases for step $i$.

\noindent
{\bf Case 1:} $rpl(T_i)=1$ and $rpl(T_{i+1})=1$.

\noindent
{\bf Case 1a:}
If $C(T_i,1)$ is the leftmost child of $T_i$
then 
we set $C(T_i,2)$  as the rightmost child of $T_i$
and $C(T_{i+1},2)$ as the lefttmost child of $T_{i+1}$
(by Lemma \ref{le:cases}(b)).

\noindent
{\bf Case 1b:}
Otherwise,
$C(T_i,1)$ is not the leftmost child of $T_i$
then 
we set $C(T_i,1)$  as the rightmost child of $T_i$
and $C(T_{i+1},1)$ as the lefttmost child of $T_{i+1}$
(by Lemma \ref{le:cases}(a)).

\noindent
{\bf Case 2:} $rpl(T_i) = 1$ and  $rpl(T_{i+1})>1$.

We have two subcases.

\noindent
{\bf Case 2a:} $T_{i+1}$ has no pony-tail.
(So $T_{i+1}$ is copying $T_i$.)

\noindent
{\bf Case 2a1:}
If $C(T_i,1)$ is the leftmost child of $T_i$
then 
we set $C(T_i,2)$  as the rightmost child of $T_i$
and $C(T_{i+1},1)$ as the leftmost child of $T_{i+1}$
(by Lemma \ref{le:cases}(d)).

\noindent
{\bf Case 2a2:}
If $C(T_i,1)$ is not the leftmost child of $T_i$
then
we set $C(T_i,1)$  as the rightmost child of $T_i$
and $C(T_{i+1},1)$ as the leftmost child of $T_{i+1}$
(by Lemma \ref{le:cases}(a)).

\noindent
{\bf Case 2b:} $T_{i+1}$ has the pony-tail and 
$T_i$ is copying $T_{i+1}$.

\noindent
{\bf Case 2b1:}
If $C(T_i,1)$ is the leftmost child of $T_i$
then 
we set $C(T_i,2)$  as the rightmost child of $T_i$
and $C(T_{i+1},2)$ as the leftmost child of $T_{i+1}$
(by Lemma \ref{le:cases}(e)).

\noindent
{\bf Case 2b2:}
If $C(T_i,1)$ is not the leftmost child of $T_i$
then
we set $C(T_i,1)$  as the rightmost child of $T_i$
and $C(T_{i+1},1)$ as the leftmost child of $T_{i+1}$
(by Lemma \ref{le:cases}(a)).

\noindent
{\bf Case 2c:} $T_{i+1}$ has the pony-tail and 
$T_{i+1}$ is copying $T_{i}$.

\noindent
{\bf Case 2c1:}
If $C(T_i,1)$ is the leftmost child of $T_i$
then 
we set $C(T_i,2)$  as the rightmost child of $T_i$
and $C(T_{i+1},1)$ as the leftmost child of $T_{i+1}$
(by Lemma \ref{le:cases}(f)).

\noindent
{\bf Case 2c2:}
If $C(T_i,1)$ is not the leftmost child of $T_i$
then
we set $C(T_i,1)$  as the rightmost child of $T_i$
and $C(T_{i+1},1)$ as the leftmost child of $T_{i+1}$
(by Lemma \ref{le:cases}(a)).

\noindent
{\bf Case 3:} $rpl(T_i) > 1$ and  $rpl(T_{i+1})=1$.

We have two subcases.

\noindent
{\bf Case 3a:} $T_{i}$ has no pony-tail.
(So $T_i$ is copying $T_{i+1}$.)

\noindent
{\bf Case 3a1:}
If $C(T_i,1)$ is the leftmost child of $T_i$
then we can prove that this case never occur, as follows.

We have set $C(T_i,1)$ as the leftmost child of $T_i$
with $rpl(T_i)>1$
in the preceding step of
either Case 2a1, 2a2, 2b2, 2c1 or 2c2.
In those cases $rpl(T_{i-1})=1$ holds,
and in Case 3a1 
$rpl(T_i)>1$ and $rpl(T_{i+1})=1$ hold and
$T_i$ has no pony-tail.
This contradicts to (co1).

\noindent
{\bf Case 3a2:}
If $C(T_i,1)$ is not the leftmost child of $T_i$
then
we set $C(T_i,1)$  as the rightmost child of $T_i$
and $C(T_{i+1},1)$ as the leftmost child of $T_{i+1}$
(by Lemma \ref{le:cases}(a)).
Set other child ordered trees of $T_i$ between 
the leftmost child and the rightmost child 
with decreasing order of $rpl$ from left to right.

\noindent
{\bf Case 3b:} $T_{i}$ has the pony-tail and
$T_{i+1}$ is copying $T_i$.

\noindent
{\bf Case 3b1:}
If $C(T_i,1)$ is the leftmost child of $T_i$
then 
we set $C(T_i,2)$  as the rightmost child of $T_i$
and $C(T_{i+1},2)$ as the leftmost child of $T_{i+1}$
(by Lemma \ref{le:cases}(e')).
Set the remaining child $C(T_{i},3)$ of $T_i$ as the middle child of $T_i$.

\noindent
{\bf Case 3b2:}
If $C(T_i,1)$ is not the leftmost child of $T_i$
then
we set $C(T_i,1)$  as the rightmost child of $T_i$
and $C(T_{i+1},1)$ as the leftmost child of $T_{i+1}$
(by Lemma \ref{le:cases}(a)).
Set the remaining child 
as the middle child of $T_i$.

\noindent
{\bf Case 3c:} $T_{i}$ has the pony-tail and
$T_i$ is copying $T_{i+1}$.

\noindent
{\bf Case 3c1:} 
$C(T_i,1)$ is the leftmost child of $T_i$.
If $T_{i+1}$ is also copying $T_{i}$
then 
we set $C(T_i,2)$  as the rightmost child of $T_i$
and $C(T_{i+1},2)$ as the leftmost child of $T_{i+1}$
(by Lemma \ref{le:cases}(c)) and
set the remaining child 
as the middle child of $T_i$.
Otherwise
one can prove that 
this case never occur.
Similar to Case 3a1.

\noindent
{\bf Case 3c2:}
If $C(T_i,1)$ is not the leftmost child of $T_i$
then
we set $C(T_i,1)$  as the rightmost child of $T_i$
and $C(T_{i+1},1)$ as the leftmost child of $T_{i+1}$
(by Lemma \ref{le:cases}(a)).
Set the remaining child 
as the middle child of $T_i$

\noindent
{\bf Case 4:} $rpl(T_i) > 1$ and  $rpl(T_{i+1})>1$.

\noindent
{\bf Case 4a:}
$C(T_i,1)$ is the leftmost child of $T_i$.

\noindent
{\bf Case 4a1:}
$rpl(T_i) \le rpl(T_{i+1})$.

We set $C(T_i,rpl(T_{i}))$  as the rightmost child of $T_i$
and $C(T_{i+1},rpl(T_{i}))$ as the leftmost child of $T_{i+1}$
(by Lemma \ref{le:cases}(g')).

Set other child ordered trees of $T_i$ between 
the leftmost child $C(T_i,1)$ and 
the rightmost child $C(T_i,rpl(T_i))$
with increasing order of $rpl$ from left to right.

\noindent
{\bf Case 4a2:}
$rpl(T_i) > rpl(T_{i+1})$.

We set $C(T_i,rpl(T_{i+1}))$  as the rightmost child of $T_i$
and $C(T_{i+1},rpl(T_{i+1}))$ as the leftmost child of $T_{i+1}$
(by Lemma \ref{le:cases}(g)).

Set other child ordered trees of $T_i$ between 
the leftmost child $C(T_i,1)$ and 
the rightmost child $C(T_i,rpl(T_{i+1}))$
with increasing order of $rpl$ from left to right.

\noindent
{\bf Case 4b:}
$C(T_i,1)$ is not the leftmost child of $T_i$.

Let $T$ be the leftmost child of $T_i$.

\noindent
{\bf Case 4b1:}
$rpl(T_i) \le rpl(T_{i+1})$.

If $rpl(T_i) < rpl(T_{i+1})$
then
we set $C(T_i,1)$  as the rightmost child of $T_i$
and $C(T_{i+1},rpl(T_i))$ as the leftmost child of $T_{i+1}$
(by Lemma \ref{le:cases}(g')).

Otherwise
$rpl(T_i) = rpl(T_{i+1})$ holds.
If $rpl(T)=rpl(T_i)$
then
we set $C(T_i,rpl(T_i)+1)$  as the rightmost child of $T_i$
and $C(T_{i+1},rpl(T_i)+1)$ as the leftmost child of $T_{i+1}$,
and if $rpl(T)\ne rpl(T_i)$
then
we set $C(T_i,rpl(T_i))$  as the rightmost child of $T_i$
and $C(T_{i+1},rpl(T_i))$ as the leftmost child of $T_{i+1}$
(by Lemma \ref{le:cases}(g)).

Set other child ordered trees of $T_i$ between 
the leftmost child  and 
the rightmost child 
with decreasing order of $rpl$ from left to right.
(Note that if $rpl(T_i)\ge 3$
then
$C(T_i,1)$ never appear at the second leftmost child of $T_i$.)

\noindent
{\bf Case 4b2:}
$rpl(T_i) > rpl(T_{i+1})$ and 
$rpl(T)\ne rpl(T_{i+1})$.

We set $C(T_i,rpl(T_{i+1}))$  as the rightmost child of $T_i$
and $C(T_{i+1},rpl(T_{i+1}))$ as the leftmost child of $T_{i+1}$
(by Lemma \ref{le:cases}(g)).
Set other child ordered trees of $T_i$ between 
the leftmost child and the rightmost child  
with decreasing order of $rpl$ from left to right.
(Note that $C(T_i,1)$ never appear at the second leftmost child of $T_i$
since $rpl(T_i)\ge 3$ holds.)

\noindent
{\bf Case 4b3:}
$rpl(T_i) > rpl(T_{i+1})$ and 
$rpl(T)= rpl(T_{i+1})$.

We show this case never occur 
in the lemma below.

\vskip 4mm

The description of the four cases 
for step $i$ is completed.

We have the following three lemmas.

\begin{lemma}\label{le:casenotoccur}
{\rm Case 4b3} never occur.
\end{lemma}
\begin{proof}
Assume for a contradiction
that the case occurs.
(In Case 4b we have defined $T$ as the leftmost child of $T_i$.)

If
$rpl(T) > 2$,
then we have set $T$ in Case 4 of the preceding setp $i-1$.
If $rpl(T_{i-1}) \le rpl(T_{i})$ and 
we set $C(T_i, rpl(T_{i-1}))$ as $T$
in either Case 4a1 or Case 4b1,
then
$rpl(T_{i-1}) = rpl(T) = rpl(T_{i+1}) < rpl(T_i)$ holds,
which contradicts to {\rm (co1)}.
Otherwise if
$rpl(T_{i-1}) = rpl(T_{i})$ and 
we set $C(T_i, rpl(T_{i}))$ or
$C(T_i, rpl(T_{i})+1)$
as $T$
in Case 4b1,
then it contradict to the condition
$rpl(T_i) > rpl(T_{i+1})$ and
$rpl(T) = rpl(T_{i+1})$
of Case 4b3.
Otherwise,
$rpl(T_{i-1}) > rpl(T_{i})$ holds,
then we set $C(T_i, rpl(T_{i}))$ as $T$
in either Case 4a2 or Case 4b2,
so 
$rpl(T) = rpl(T_{i})$ holds, 
which contradicts to Case 4b3.

If
$rpl(T)=2$,
then we set  $T$ in either Case 2b1, 4a1, 4a2, 4b1 or 4b2
of the preceding step $i-1$.
If we set $T$ in Case 2b1 then
$T_{i}$ has the pony-tail and $rpl(T_i)=2$,
which contradicts to $rpl(T_i) > rpl(T_{i+1})>1$.
If we set $T$ in Case 4a1 or Case 4b1 then 
$rpl(T_{i-1})=rpl(T)=rpl(T_{i+1})<rpl(T_i)$,
which contradicts to {\rm (co1)}.
If we set $T$ in either Case 4a2 or Case 4b2 then 
$rpl(T_{i-1})>rpl(T_{i})=rpl(T) $
which contradicts
to Case 4b3.
\end{proof}

\begin{lemma}\label{le:copy}
(a)
If $rpl(T)=1$, $T'$ has no pony-tail
and $T'$ is copying $T$,
then
$C(T',1)$ is copying $C(T,2)$.

\noindent
(b)
If $rpl(T)=1$, $T'$ has the pony-tail and
$T'$ is copying $T$,
then
$C(T',1)$ is copying $C(T,2)$.


\end{lemma}
\begin{proof} 
(Sketch.) See Fig. \ref{fig:copy}.
\end{proof}
We need above lemma in the proof of the next lemma.

\begin{figure}[htb]
  \begin{center}
     \includegraphics[width=14cm,pagebox=cropbox,clip]{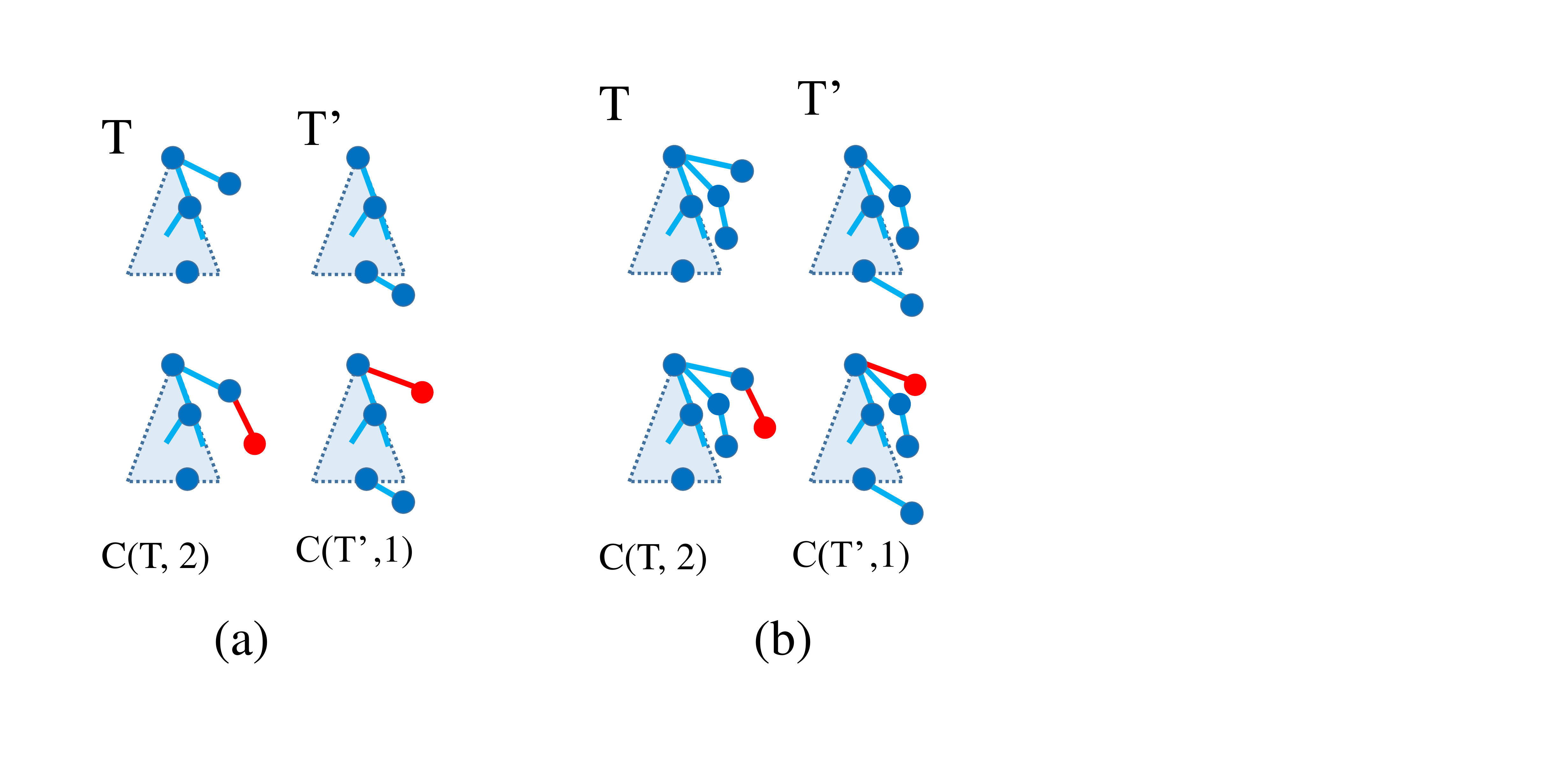}
  \end{center}
  \caption{
     Illustrations for Lemma \ref{le:copy}.  }
  \label{fig:copy}
\end{figure}


\begin{lemma}\label{le:inv}
Assume that {\rm (co1)} is satisfied.
If {\rm (co2)} is satisfied for $i=1,2,\cdots,s$
then, after executing step $i=s$,
{\rm (co2)} is satisfied for $i=s+1$. 
\end{lemma}
\begin{proof} 
{\rm\bf First part of (co2)}
We have the following three cases to consider.
For each case we can prove {\rm (co2)} is satisfied for $i=s+1$,
as follows.

\noindent{\bf Case 1:}
$T'_{u'-1}$ is the rightmost child of $T_{s-1}$,
$T'_{u'}$ is the lefhtmost child of $T_{s}$ and 
$T'_{u'+1}$ is the second lefhtmost child of $T_{s}$.

If those three ordered trees violate {\rm (co2)}
then $rpl(T'_{u'-1}) = rpl(T'_{u'+1}) =1 < rpl(T'_{u'})$ holds.

Only Case 4b1 
with $rpl(T_{s-1})<rpl(T_s)$
sets $T'_{u'-1}$ and $T'_{u'}$
so that 
$rpl(T'_{u'-1})=1 < rpl(T'_{u'})$.
If so  $rpl(T_s)\ge 3$ holds.
However no case set 
(the second leftmost child of $T_s$)
$T'_{u'+1}$ with $rpl(T'_{u'+1})=1$
since if $rpl(T_s)\ge 3$ 
then no case set 
$C(T_s,1)$ as the second leftmost child of $T_s$.
Thus {\rm (co2)} is satisfied.

\noindent{\bf Case 2:}
$T'_{u'-1}$,
$T'_{u'}$ and 
$T'_{u'+1}$ are children of $T_{s}$.

Those three ordered trees never violate {\rm (co2)}
since they are children of $T_{s}$ and have distinct $rpl$'s.


\noindent{\bf Case 3:}
$T'_{u'-1}$ is the second rightmost child of $T_{s-1}$,
$T'_{u'}$ is the rightmost child of $T_{s-1}$ and 
$T'_{u'+1}$ is the leftmost child of $T_{s}$.

If those three ordered trees violate {\rm (co2)}
then $rpl(T'_{u'-1}) = rpl(T'_{u'+1}) =1 < rpl(T'_{u'})$ holds.
This occurs only 
when we set $T'_{u'}$ and $T'_{u'+1}$
in either Case 2a1 or Case 2c1.
For those cases $rpl(T_{s-1})=1$ holds,
and
$rpl(T'_{u'-1}) = rpl(T'_{u'+1})=1$,
$T'_{u'}$ has the pony-tail and 
$T'_{u'+1}$ is copying $T'_{u'}$
by Lemma \ref{le:copy}(a) and (b).
Thus {\rm (co2)} is satisfied.

\vskip 4mm

\noindent
{\rm\bf Second part of (co2)}
%
If $T'_{u'-1}, T'_{u'}$ and $T'_{u'+1}$ are siblings,
since each child ordered tree has a distinct $rpl$,
the claim is satisfied.
So assume otherwise, that is
$T'_{u'-1}$ and $T'_{u'+1}$ are not siblings.
We have the following two cases.

\noindent{\bf Case 1:}
$T'_{u'}$ and $T'_{u'+1}$ are not siblings.

Now $T'_{u'-1}$ and $T'_{u'}$ are siblings.
If $T'_{u'-1},T'_{u'},T'_{u'+1}$ violate (co2)
then $2\le rpl(T'_{u'-1})<rpl(T'_{u'})$ and
$rpl(T'_{u'})>rpl(T'_{u'+1})\ge 2$ hold.
No case set $T'_{u'}$ and $T'_{u'+1}$
with $rpl(T'_{u'})>rpl(T'_{u'+1})\ge 2$.
Thus this case never occur.

\noindent{\bf Case 2:}
$T'_{u'-1}$ and $T'_{u'}$ are not siblings.

Now $T'_{u'}$ and $T'_{u'+1}$ are siblings.
If $T'_{u'-1},T'_{u'},T'_{u'+1}$ violate (co2)
then $2\le rpl(T'_{u'+1})<rpl(T'_{u'})$ and
$rpl(T'_{u'})>rpl(T'_{u'-1})\ge 2$ hold.
No case set $T'_{u'-1}$ and $T'_{u'}$
with $rpl(T'_{u'})>rpl(T'_{u'-1})\ge 2$.
Thus this case never occur.
\end{proof}

Now we have the following theorem.

\begin{theorem}\label{th:main}
There is a Gray code for ordered trees with $n$ vertices
such that each ordered tree is derived from the preceding ordered tree
by removing a leaf then appending a leaf.
\end{theorem}

By constructing the necessary part of $F^O_n$ on the fly
one can generate each ordered tree in a Gray code for $S_n$
in $O(n^2)$ time for each ordered tree.

\section{Conclusion}

In this paper 
we have designed a Gray code for ordered trees with $n$ vertices
such that
each ordered tree is derived from the preceding ordered tree
by removing a leaf then appending a leaf.

Can we design a Gray code for binary trees with $n$ vertices
such that
each binary tree is derived from the preceding binary tree by
removing a leaf then appending a leaf?

\bibliography{main}
\bibliographystyle{plainurl}

\end{document}